\documentclass[11pt]{article}
\usepackage{amsmath,amssymb,amsthm}
\AtBeginDvi{}
\usepackage[dvipsnames]{xcolor}
\usepackage[dvipdfm,bookmarks=true,bookmarksnumbered=true,colorlinks=true]{hyperref}
\usepackage{txfonts}
\newtheorem{prop}{Proposition}
\newtheorem{thm}[prop]{Theorem}

\def\theequation{\thesection.\arabic{equation}}
\makeatletter
\@addtoreset{equation}{section}
\makeatother

\title{Toda lattice hierarchy and 
Goldstein-Petrich flows for plane curves}
\author{Kenji Kajiwara$^1$ and Saburo Kakei$^2$\\[2mm]
{}$^1$ 
{\small Institute of Mathematics for Industry, Kyushu University,}\\
{\small 744 Motooka, Fukuoka 819-0395, Japan.}\\[1mm]
{}$^2$ {\small Department of Mathematics, Rikkyo University,}\\
{\small 3-34-1 Nishi-ikebukuro, Toshima-ku, Tokyo 171-8501, Japan.}
}
\date{}
\begin{document}
\maketitle
\begin{abstract}
A relation between the Goldstein-Petrich hierarchy for plane curves
and the Toda lattice hierarchy is investigated. A representation formula 
for plane curves is given in terms of a special class of 
$\tau$-functions of the Toda lattice hierarchy.
A representation formula for discretized plane curves is also discussed.
\end{abstract}
\section{Introduction}
Intimate connection between integrable systems and 
differential geometry of curves and surfaces has been
important topic of intense research \cite{BobenkoSuris,RogersSchief}.
Goldstein and Petrich introduced a hierarchy of commuting flows 
for plane curves that is related to the modified Korteweg-de Vries 
(mKdV) hierarchy \cite{GP}. 
The second Goldstein-Petrich flow is defined by
the modified Korteweg-de Vries equation,
\begin{equation}
\frac{\partial\kappa}{\partial t}
=\frac{\partial^3 \kappa}{\partial x^3}+\frac{3}{2}\kappa^2
\frac{\partial\kappa}{\partial x}, 
\label{mKdV_kappa}
\end{equation}
where $\kappa=\kappa(x,t)$ denotes the curvature and $x$ is the
arc-length. This result has been
extended and investigated from various viewpoints
\cite{ChouQu,DoliwaSantini,FK,IKMO1,IKMO2,Ivey,LangerPerline,Musso,NSW}.  
In \cite{IKMO1,IKMO2}, a
representation formula for curve motion in terms of the $\tau$ function with
respect to the second Goldstein-Petrich flow has been presented by means
of the Hirota bilinear formulation and determinant expression of
solutions. The aim of this article is to
generalize the results in \cite{IKMO1,IKMO2} to the whole hierarchy.  
We will show how the Goldstein-Petrich hierarchy is embedded in the 
Toda lattice hierarchy\cite{TakasakiBook,UenoTakasaki}.
We remark that the semi-discrete case, discussed in \cite{IKMO2}, 
is not considered in this paper.

An advantage of infinite hierarchical formulation is its relation 
to integrable discretization. Miwa showed that Hirota's 
discrete Toda equation \cite{Hirota_DAGTE} can be obtained by 
applying a change of coordinate to the KP hierarchy
\cite{JimboMiwa,Miwa,TakasakiBook}.
Using a generalization of Miwa's approach, we will show that 
Matsuura's discretized curve motion \cite{Matsuura} can be 
obtained also from the Toda lattice hierarchy.
Another merit of the KP theoretic formulation is 
Lie algebraic aspect of the hierarchy \cite{JimboMiwa,MJD}. 
We will discuss a relationship between 
the Goldstein-Petrich hierarchy and a real form of the affine Lie algebra 
$\widehat{\mathfrak{sl}}(2,\mathbb{C})$.

\section{Goldstein-Petrich flows for Euclidean plane curves}
We assume that $\mathbf{r}(x)={}^t\!\left(X(x), Y(x)\right)$
is a curve in Euclidean plane $\mathbb{R}^2$, parameterized by
the arc-length $x$. Define the tangent vector $\hat{\mathbf{t}}$ and 
the unit normal $\hat{\mathbf{n}}$ by 
\begin{equation}
\hat{\mathbf{t}}=\mathbf{r}_x, \quad
\hat{\mathbf{n}}=\begin{bmatrix}0 & -1\\ 1 & 0\end{bmatrix}\hat{\mathbf{t}}.
\label{def:t,n}
\end{equation}
Here the subscript $x$ indicates differentiation.
The Frenet equation for $\mathbf{r}$ is given by
\begin{equation}
\hat{\mathbf{t}}_x = \kappa \hat{\mathbf{n}},\quad
\hat{\mathbf{n}}_x = -\kappa \hat{\mathbf{t}},
\label{FrenetEq}
\end{equation}
where $\kappa$ is the curvature of the curve $\mathbf{r}$.
Goldstein and Petrich \cite{GP} considered 
dynamics of a plane curve described by the 
equation of the form
\begin{equation}
\frac{\partial\mathbf{r}}{\partial t_n}
 = f^{(n)} \hat{\mathbf{n}}+g^{(n)} \hat{\mathbf{t}}. 
\label{eq:time-evolutions}
\end{equation}
The coefficients $f^{(n)}=f^{(n)}(x,t)$, $g^{(n)}=g^{(n)}(x,t)$ 
($t=(t_1,t_2,t_3,\ldots)$) are differential polynomials in $\kappa$.
We remark that our choice of signature in \eqref{FrenetEq} is 
different from that of \cite{GP}. 
Following the discussion in \cite{GP}, 
we choose $f^{(n)}(x,t)$, $g^{(n)}(x,t)$ as 
\begin{equation}
\begin{aligned}
& f^{(1)} = 0, \quad g^{(1)}=1,\quad 
  f^{(2)} = \kappa_x, \quad g^{(2)}=\kappa^2/2,\\
& g^{(n)}_x = \kappa f^{(n)},\quad
f^{(n+1)}=\left(f^{(n)}_x+\kappa g^{(n)}\right)_x.
\end{aligned}
\label{def:coeff_f_g}
\end{equation}
We call as Goldstein-Petrich hierarchy the equations 
defined by \eqref{def:t,n}, \eqref{FrenetEq},
\eqref{eq:time-evolutions} and \eqref{def:coeff_f_g}.

Applying the condition \eqref{def:coeff_f_g} to
\eqref{eq:time-evolutions}, we obtain
\begin{equation}
\frac{\partial\hat{\mathbf{t}}}{\partial t_n}
= \left(f^{(n)}_x+\kappa g^{(n)}\right)\hat{\mathbf{n}},\quad
\frac{\partial\hat{\mathbf{n}}}{\partial t_n}
= -\left(f^{(n)}_x+\kappa g^{(n)}\right)\hat{\mathbf{t}}.
\label{eq:dt/dtn,dn/dtn}
\end{equation}
The compatibility condition for 
\eqref{FrenetEq} and \eqref{eq:dt/dtn,dn/dtn} is reduced to 
\begin{equation}
\frac{\partial\kappa}{\partial t_n}
=\left(f^{(n)}_x+\kappa g^{(n)}\right)_x=f^{(n+1)}.
\label{eq:dkappa/dtn}
\end{equation}
The case $n=2$ of \eqref{eq:dkappa/dtn} gives 
the mKdV equation \eqref{mKdV_kappa}.
One finds that
\begin{equation}
f^{(n)}= \Omega f^{(n-1)}, \quad
\Omega = 
\partial_x^2+\kappa^2+\kappa_x\partial_x^{-1}\circ \kappa.
\label{def:recursion_operator_omega}
\end{equation}
We remark that the operator $\Omega$ is the recursion operator for the
modified KdV hierarchy \cite{ChernPeng}. 

We now introduce complex coordinate via a map 
$\rho:\mathbb{R}^2\to\mathbb{C}$ given by
\begin{equation}
\rho(X,Y) = X+\sqrt{-1}\,Y.
\end{equation}
and define $Z$, $T$, $N$ as
\begin{equation}
Z= \rho(\mathbf{r}), \quad
T= \rho(\hat{\mathbf{t}}), \quad
N= \rho(\hat{\mathbf{n}}) = \sqrt{-1}\, T.
\end{equation}
Since $|\hat{\mathbf{t}}|=|\hat{\mathbf{n}}|=1$, 
the complex variables $T$ and $N$ satisfy $|T|=|N|=1$.
The equations 
\eqref{def:t,n}, \eqref{FrenetEq}, \eqref{eq:time-evolutions} 
are rewritten as
\begin{equation}
T=Z_x,\quad T_x = \sqrt{-1}\kappa T,\quad
\frac{\partial Z}{\partial t_n}= \left(g^{(n)}+\sqrt{-1}\,f^{(n)}\right)T.
\label{eqs:GPhier_cplx_coord}
\end{equation}

\section{Toda lattice hierarchy}
In this section, we briefly review the theory of Toda lattice hierarchy 
using the language of difference operators
\cite{TakasakiBook,UenoTakasaki} (See also \cite{Kaji,Takebe1,Takebe2}). 
We denote as $e^{\partial_s}$ the shift operator with respect to $s$: 
$e^{\partial_s}f(s)=f(s+1)$.
For a difference operator 
$A(s)=\sum_{-\infty<j<+\infty}a_j(s)e^{j\partial_s}$, 
we define the non-negative and negative part of $A(s)$ as 
\begin{equation}
\left(A(s)\right)_{\geq 0}=
\sum_{0\leq j<+\infty}a_j(s)e^{j\partial_s},
\quad
\left(A(s)\right)_{<0}=
\sum_{-\infty<j<0}a_j(s)e^{j\partial_s}.
\end{equation}
Let $L^{(\infty)}(s)$, $L^{(0)}(s)$ be difference operators of the form
\begin{equation}
L^{(\infty)}(s)
= e^{\partial_s} + \sum_{-\infty<j\leq 0} b_j(s)e^{j\partial_s},
\quad
L^{(0)}(s)= \sum_{-1\leq j< +\infty} c_j(s)e^{j\partial_s},
\label{def:L}
\end{equation}
where we assume $c_{-1}(s)\neq 0$ for any $s$. 
We introduce two sets of infinitely many variables 
$x=(x_1,x_2,\ldots)$, $y=(y_1,y_2,\ldots)$ and 
define the weight of the variables as 
\begin{equation}
\mathrm{weight}(x_n)=n, \quad \mathrm{weight}(y_n)=-n \quad (n=1,2,\ldots).
\end{equation}
Each coefficient of $L^{(\infty)}(s)$, $L^{(0)}(s)$ is a function
of $x,y$, i.e. $b_j(s)=b_j(s;x,y)$, $c_j(s)=c_j(s;x,y)$.
The Toda lattice hierarchy is defined as the following set of 
differential equations of Lax-type:
\begin{align}
&\begin{aligned}
&\frac{\partial L^{(\infty)}(s)}{\partial x_n}
=\left[B_n(s),\,L^{(\infty)}(s)\right],\quad
\frac{\partial L^{(0)}(s)}{\partial x_n}
=\left[B_n(s),\,L^{(0)}(s)\right],\\
&B_n(s)=\left(L^{(\infty)}(s)^n\right)_{\geq 0}
\quad (n=1,2,3,\dots),
\end{aligned}
\label{def:TLhierarchy_x}\\
&\begin{aligned}
&\frac{\partial L^{(\infty)}(s)}{\partial y_n}
=\left[C_n(s),\,L^{(\infty)}(s)\right],\quad
\frac{\partial L^{(0)}(s)}{\partial y_n}
=\left[C_n(s),\,L^{(0)}(s)\right],\\
&C_n(s)=\left(L^{(0)}(s)^n\right)_{<0}
\quad (n=1,2,3,\dots).
\end{aligned}
\label{def:TLhierarchy_y}
\end{align}
\begin{prop}[\cite{UenoTakasaki}, Proposition 1.4]
\label{Prop:SatoEqs}
Let $L^{(\infty)}$, $L^{(0)}$ be difference operators of the form \eqref{def:L} 
and satisfy the differential equations \eqref{def:TLhierarchy_x}, 
\eqref{def:TLhierarchy_y}.
Then there exist difference operators 
$\hat{W}^{(\infty)}(s)$, $\hat{W}^{(0)}(s)$ of the form,
\begin{equation}
\begin{aligned}
\hat{W}^{(\infty)}(s)&=1+\sum_{j=1}^{\infty}
\hat{w}^{(\infty)}_j(s)e^{-j\partial_s},\\
\hat{W}^{(0)}(s)&=\sum_{j=0}^{\infty}\hat{w}_j^{(0)}(s)e^{j\partial_s}
\quad (\hat{w}_0^{(0)}(s)\neq 0),
\end{aligned}
\label{def:WinftyW0}
\end{equation}
satisfying the following equations:
\begin{equation}
\begin{aligned}
L^{(\infty)}(s)&=
\hat{W}^{(\infty)}(s)e^{\partial_s}\hat{W}^{(\infty)}(s)^{-1},
\\
L^{(0)}(s)&=\hat{W}^{(0)}(s)e^{-\partial_s}\hat{W}^{(0)}(s)^{-1},
\end{aligned}
\label{rel:LtoW}
\end{equation}
\begin{equation}
\begin{aligned}
\frac{\partial\hat{W}^{(\infty)}(s)}{\partial x_n} &=
B_n(s)\hat{W}^{(\infty)}(s)-\hat{W}^{(\infty)}(s)e^{n\partial_s},
\\
\frac{\partial\hat{W}^{(\infty)}(s)}{\partial y_n} &=
C_n(s)\hat{W}^{(\infty)}(s),
\\
\frac{\partial\hat{W}^{(0)}(s)}{\partial x_n} &=
B_n(s)\hat{W}^{(0)}(s),
\\
\frac{\partial\hat{W}^{(0)}(s)}{\partial y_n} &=
C_n(s)\hat{W}^{(0)}(s)-\hat{W}^{(0)}(s)e^{-n\partial_s}.
\end{aligned}
\label{SatoEqs}
\end{equation}
\end{prop}
\begin{prop}[\cite{UenoTakasaki}, (1.2.18)] 
The difference operators 
$\hat{W}^{(\infty)}(s)$, $\hat{W}^{(0)}(s)$ in Proposition
\ref{Prop:SatoEqs} satisfy 
\begin{equation}
\begin{aligned}
&\hat{W}^{(\infty)}(s;x',y')\exp\!\left[
\sum_{n=1}^{\infty}(x'_n-x_n)e^{n\partial_s}
\right]
\hat{W}^{(\infty)}(s;x,y)^{-1}\\
&\qquad =
\hat{W}^{(0)}(s;x',y')\exp\!\left[
\sum_{n=1}^{\infty}(y'_n-y_n)e^{-n\partial_s}
\right]
\hat{W}^{(0)}(s;x,y)^{-1}
\end{aligned}
\label{WexpW=WexpW}
\end{equation}
for any $x$, $x'$, $y$, $y'$ and any integer $s$.
\end{prop}

Define $\hat{w}^{(\infty)*}_j(s;x,y)$, $\hat{w}^{(0)*}_j(s;x,y)$ 
by expanding 
$\hat{W}^{(\infty)}(s;x,y)^{-1}$, $\hat{W}^{(0)}(s;x,y)^{-1}$ 
with respect to $e^{\partial_s}$:
\begin{equation}
\begin{aligned}
\hat{W}^{(\infty)}_j(s;x,y)^{-1}
&=\sum_{j=0}^{\infty}e^{-j\partial_s}\hat{w}^{(\infty)*}_j(s+1;x,y),
\\
\hat{W}^{(0)}(s;x,y)^{-1}
&=\sum_{j=0}^{\infty}e^{j\partial_s}\hat{w}^{(0)*}_j(s+1;x,y).
\end{aligned}
\label{def:WinftyW0_inverse}
\end{equation}
{}From \eqref{def:WinftyW0}, \eqref{rel:LtoW} and 
\eqref{def:WinftyW0_inverse}, we obtain
\begin{align}
b_0(s)&= \hat{w}^{(\infty)}_{1}(s)+\hat{w}^{(\infty)*}_{1}(s+1)
=\hat{w}^{(\infty)}_{1}(s)-\hat{w}^{(\infty)}_{1}(s+1),
\nonumber\\
b_{-n}(s)&= \hat{w}^{(\infty)}_{n+1}(s)+
\hat{w}^{(\infty)*}_{n+1}(s+1-n)+
\sum_{j=1}^{n}\hat{w}^{(\infty)}_{j}(s)
\hat{w}^{(\infty)*}_{n+1-j}(s+1-n) \quad (n\geq 1),
\nonumber\\
c_n(s)&= \sum_{j=0}^{n+1}
\hat{w}^{(0)}_j(s)\hat{w}^{(0)*}_{n-j+1}(s+n+1)
\quad (n\geq -1).
\label{rel:cn_to_w}
\end{align}

\begin{thm}[\cite{UenoTakasaki}, Theorem 1.7]
There exists a function $\tau(s)=\tau(s;x,y)$ satisfying
\begin{equation}
\begin{aligned}
\hat{w}^{(\infty)}_j(s;x,y) &=
 \frac{p_j(-\tilde{\partial}_x)\tau(s;x,y)}{\tau(s;x,y)},
\\
\hat{w}^{(0)}_j(s;x,y) &=
\frac{p_j(-\tilde{\partial}_y)\tau(s+1;x,y)}{\tau(s;x,y)},
\\
\hat{w}^{(\infty)*}_j(s;x,y) &= 
\frac{p_j(\tilde{\partial}_x)\tau(s;x,y)}{\tau(s;x,y)},
\\
\hat{w}^{(0)*}_j(s;x,y) &= 
\frac{p_j(\tilde{\partial}_y)\tau(s-1;x,y)}{\tau(s;x,y)}
\end{aligned}
\label{def:tau-function}
\end{equation}
where $\tilde{\partial}_x=
\left(\partial_{x_1},\partial_{x_2}/2,\partial_{x_3}/3,\ldots\right)$, 
$\tilde{\partial}_y=\left(
\partial_{y_1},\partial_{y_2}/2,\partial_{y_3}/3,\ldots\right)$, and 
the polynomials $p_n(t)$ ($n=0,1,2,\ldots$) are defined by
\begin{equation}
\xi(t,\lambda)
=\exp\!\left[\sum_{j=1}^{\infty}t_n\lambda^j\right]=\sum_{n=0}^{\infty}p_n(t)\lambda^n,
\quad t=(t_1,t_2,\ldots).
\end{equation}
Furthermore, the $\tau$-function $\tau(s;x,y)$ of the 
Toda lattice hierarchy is determined uniquely by
\eqref{def:tau-function} up to a constant multiple factor.
\end{thm}

It follows that 
\begin{equation}
\begin{aligned}
c_{-1}(s) &= 
\hat{w}^{(0)}_0(s)\hat{w}^{(0)*}_0(s)=\frac{\tau(s+1)\tau(s-1)}{\tau(s)^2},\\
c_0(s) &= 
\hat{w}^{(0)}_0(s)\hat{w}^{(0)*}_1(s+1)+
\hat{w}^{(0)}_1(s)\hat{w}^{(0)*}_0(s+1)
= \frac{\partial}{\partial y_1}\log\frac{\tau(s)}{\tau(s+1)}.
\end{aligned}
\label{rel:c{-1}c0_to_tau}
\end{equation}

\begin{thm}[\cite{UenoTakasaki}, Theorem 1.11]
$\tau$-functions of Toda lattice hierarchy satisfy the following 
equation (bilinear identity):
\begin{equation}
\begin{aligned}
&\oint\tau(s';x'-[\lambda^{-1}],y')\tau(s;x+[\lambda^{-1}],y)
e^{\xi(x'-x,\lambda)}\lambda^{s'-s}d\lambda
\\
=&
\oint\tau(s'+1;x',y'-[\lambda])\tau(s-1;x,y+[\lambda])
e^{\xi(y'-y,\lambda^{-1})}\lambda^{s'-s}d\lambda,
\end{aligned}
\label{BilinearIdentity}
\end{equation}
where $[\lambda]=\left(\lambda,\lambda^2/2,\lambda^3/3,\ldots\right)$, 
and we have used the notation of formal residue,
\begin{equation}
\oint\left(\sum_n a_n\lambda^n\right)d\lambda
=2\pi\sqrt{-1}\,a_{-1}.
\end{equation}
Conversely, if $\tau(s;x,y)$ solves the bilinear identity 
\eqref{BilinearIdentity}, then $\hat{W}^{(\infty)}(s;x,y)$ and 
$\hat{W}^{(0)}(s;x,y)$ defined by \eqref{def:WinftyW0} and 
\eqref{def:tau-function} satisfy \eqref{SatoEqs}.
\end{thm}

\section{Time-flows with negative weight with 2-reduction condition}
\subsection{Reduction to Goldstein-Petrich hierarchy}
We now impose the 2-reduction condition\cite{UenoTakasaki}
\begin{equation}
L^{(\infty)}(s)^2=e^{2\partial_s},\quad
L^{(0)}(s)^2=e^{-2\partial_s},
\label{def:2-reduction}
\end{equation}
that implies
\begin{gather}
W^{(\infty)}(s+2)=W^{(\infty)}(s),\quad
W^{(0)}(s+2)=W^{(0)}(s),
\label{2-reduced_W}\\
L^{(\infty)}(s+2)=L^{(\infty)}(s),\quad
L^{(0)}(s+2)=L^{(0)}(s).
\end{gather}

\begin{prop}[\cite{UenoTakasaki}, Proposition 1.13]
Let $L^{(\infty)}(s;x,y)$, $L^{(0)}(s;x,y)$ be solutions to 
the Toda lattice hierarchy \eqref{def:TLhierarchy_x},
\eqref{def:TLhierarchy_y}, which 
satisfy the 2-reduction conditions \eqref{def:2-reduction}. 
Then one finds that
\begin{equation}
\begin{aligned}
\frac{\partial L^{(\infty)}}{\partial x_{2n}}
=\frac{\partial L^{(0)}}{\partial x_{2n}}
=\frac{\partial L^{(\infty)}}{\partial y_{2n}}
=\frac{\partial L^{(0)}}{\partial y_{2n}}=0
\end{aligned}
\end{equation}
for $n=1,2,\ldots$.
\end{prop}
\begin{prop}[\cite{UenoTakasaki}, Corollary 1.14] 
Suppose $L^{(\infty)}(s;x,y)$, $L^{(0)}(s;x,y)$ be solutions to 
the Toda lattice hierarchy \eqref{def:TLhierarchy_x},
\eqref{def:TLhierarchy_y}, which 
satisfy the 2-reduction conditions \eqref{def:2-reduction}. 
Then there exist suitable difference operators 
$\hat{W}^{(\infty)}(s;x,y)$, $\hat{W}^{(0)}(s;x,y)$ 
such that the corresponding $\tau$ functions subject to the 
following conditions:
\begin{equation}
\begin{aligned}
& \tau(s;x,y)=\tau'(s;x,y)
\exp\!\left(-\sum_{n=1}^{\infty}nx_n y_n\right),\\
& \tau'(s+2;x,y)=\tau'(s;x,y),\\
& \frac{\partial\tau'(s;x,y)}{\partial x_{2n}}
=\frac{\partial\tau'(s;x,y)}{\partial y_{2n}}=0 \quad (n=1,2,\ldots).
\end{aligned}
\end{equation}
\end{prop}

We consider the time-evolutions with respect to 
the variables with negative weight $y=(y_1,y_2,\ldots)$ 
under the 2-reduction condition \eqref{def:2-reduction}. 
In this case, one can write down
the difference operators $C_n(s)$ ($n=1,2,\ldots$) explicitly:
\begin{equation}
C_{2n}(s)=e^{-2n\partial_s},\quad
C_{2n-1}=\sum_{j=-1}^{2n-3}c_j(s)e^{(j-2n)\partial_s}.
\label{2-reduced_Cn}
\end{equation}
Applying \eqref{2-reduced_Cn} to 
\eqref{def:TLhierarchy_y} and \eqref{SatoEqs}, 
we obtain the following equations 
($n=0,1,2,\ldots$):
\begin{align}
\frac{\partial w_1(s)}{\partial y_{2n+1}} &= c_{2n-1}(s) 
\label{dw1/dy(2n+1)}\\
\frac{\partial c_{2n}(s)}{\partial y_1} &=
\frac{\partial c_{0}(s)}{\partial y_{2n+1}} =
c_{-1}(s)c_{2n+1}(s+1)-c_{-1}(s+1)c_{2n+1}(s),
\label{dc2n/y1}\\
\frac{\partial c_{-1}(s)}{\partial y_{2n+1}} &=
\frac{\partial c_{2n-1}(s)}{\partial y_1} =
c_{-1}(s)\left\{c_{2n}(s+1)-c_{2n}(s)\right\},
\label{dc(-1)/dy(2n+1)}
\end{align}
where we have used the property $c_j(s+2)=c_j(s)$.
\begin{prop}
For $n=0,1,2,\ldots$, the coefficients $c_{n}(s;x,y)$ can be 
represented by $c_{-1}(s;x,y)$. For example, 
$c_0(s;x,y)$ and $c_1(s;x,y)$ can be written as
\begin{equation}
\begin{aligned}
c_0(s)&=-\frac{1}{2c_{-1}(s)}\frac{\partial c_{-1}(s)}{\partial y_1}
=-\frac{1}{2}\frac{\partial}{\partial y_1}\log c_{-1}(s),
\\
c_1(s)&=-\frac{c_{-1}(s)}{2}\left\{
c_0(s)^2+\frac{\partial c_0(s)}{\partial y_1}
\right\}
\\
&=-\frac{c_{-1}(s)}{8}\left[
\left\{
\frac{\partial}{\partial y_1}\log c_{-1}(s)
\right\}^2
 -2\frac{\partial^2}{\partial y_1^2}\log c_{-1}(s)
\right].
\end{aligned}
\end{equation}
\end{prop}
\begin{proof}
{}From \eqref{def:L} and \eqref{def:2-reduction}, we have
\begin{equation}
\begin{aligned}
& c_{-1}(s)c_{-1}(s-1)=1,\quad c_0(s)+c_0(s-1)=0,\\
& c_{-1}(s)c_{k+1}(s-1)+c_{-1}(s+k+1)c_{k+1}(s)
+\sum_{j=0}^k c_j(s)c_{k-j}(s+j)=0.
\end{aligned}
\label{2-reduced_cn}
\end{equation}
The desired result can be obtained from 
\eqref{dc2n/y1}, \eqref{dc(-1)/dy(2n+1)} and \eqref{2-reduced_cn}.
\end{proof}

\noindent\textbf{Remark:} 
Under the 2-reduction conditions \eqref{def:2-reduction}, 
the map
\begin{equation}
\sum_{n\in\mathbb{Z}}a_n(s)e^{n\partial_s}
\mapsto 
\sum_{s\in\mathbb{Z}}
\begin{bmatrix}
a_n(0) & 0\\ 0 & a_n(1)
\end{bmatrix}
\begin{bmatrix}
0 & 1\\ \zeta^2 & 0
\end{bmatrix}^n
\end{equation}
gives an algebra isomorphism \cite{UenoTakasaki}. 
For example, the operators 
$C_1(s)$, $C_3(s)$ are mapped as follows:
\begin{equation}
\begin{aligned}
C_1(s) & \mapsto 
\begin{bmatrix}
c_{-1}(0) & 0\\ 0 & c_{-1}(1)
\end{bmatrix}
\begin{bmatrix}
0 & \zeta^{-2} \\ 1& 0
\end{bmatrix}
=\begin{bmatrix}
0 & c_{-1}(0)\zeta^{-2}\\ 1/c_{-1}(0) & 0
\end{bmatrix},\\
C_3(s) & \mapsto 
\begin{bmatrix}
c_0(0)\zeta^{-2} & c_{-1}(0)\zeta^{-4}+c_1(0)\zeta^{-2}\\
\zeta^{-2}/c_{-1}(0)+c_1(1) & -c_0(0)\zeta^{-2}
\end{bmatrix}.
\end{aligned}
\end{equation}
Applying this isomorphism to the equations \eqref{def:TLhierarchy_x}, 
\eqref{def:TLhierarchy_y}, one obtains the Lax equations of $2\times
2$-matrix form. 
\medskip

For $n=0,1,2,\ldots$, define $F^{(n)}(s)$ and $G^{(n)}(s)$ as
\begin{equation}
\begin{aligned}
F^{(n)}(s) &= \frac{1}{2}\left\{
c_{-1}(s+1)c_{2n-1}(s)-c_{-1}(s)c_{2n-1}(s+1)
\right\},\\
G^{(n)}(s) &= \frac{1}{2}\left\{
c_{-1}(s+1)c_{2n-1}(s)+c_{-1}(s)c_{2n-1}(s+1)
\right\}.
\end{aligned}
\label{def:FG}
\end{equation}
{}From \eqref{dw1/dy(2n+1)}, \eqref{2-reduced_cn} and \eqref{def:FG},
we have
\begin{equation}
\frac{\partial w_1(s)}{\partial y_{2n+1}} %= c_{2n-1}(s)
=\frac{F^{(n)}(s)+G^{(n)}(s)}{c_{-1}(s+1)}
=c_{-1}(s)\left\{
F^{(n)}(s)+G^{(n)}(s)\right\}.
\label{dw1/dy_2n+1}
\end{equation}
It is straightforward to show that
\begin{equation}
 \begin{aligned}
\frac{\partial F^{(n)}(s)}{\partial y_1}&=
2c_0(s)G^{(n)}(s)+c_{2n}(s+1)-c_{2n}(s),\\
\frac{\partial G^{(n)}(s)}{\partial y_1}&=2c_0(s)F^{(n)}(s).
 \end{aligned}
\label{dF/dy1,dG/dy1}
\end{equation}

Next we consider reality condition. Assume 
$x_j,y_j\in\mathbb{R}$ ($j=1,2,\ldots$) and 
the $\tau$-function $\tau(s;x,y)$ satisfies
\begin{equation}
\overline{\tau(s;x,y)}=\tau(s+1;x,y)
\label{reality_condition},
\end{equation}
where $\overline{\,\cdot\,}$ denotes complex conjugation.
Under this condition, the following relations hold:
\begin{equation}
\begin{aligned}
& \overline{\hat{w}^{(\infty)}_j(s)}=\hat{w}^{(\infty)}_j(s+1), 
\quad 
\overline{\hat{w}^{(0)}_j(s)}=\hat{w}^{(0)}_j(s+1), \\
& \overline{b_{-n}(s)}= b_{-n}(s+1),\quad
\overline{c_{n}(s)}= c_{n}(s+1),\\
& \overline{F^{(n)}(s)}=-F^{(n)}(s), \quad
\overline{G^{(n)}(s)}=G^{(n)}(s).
\end{aligned}
\end{equation}
Furthermore, it follows from \eqref{rel:c{-1}c0_to_tau} that
\begin{equation}
c_{-1}(s)\,\overline{c_{-1}(s)}=1,\quad
c_0(s)+\overline{c_0(s)}=0.
\end{equation}
\begin{thm}[Representation formula in terms of the $\tau$-functions]
If we set 
\begin{equation}
\begin{aligned}
&x=2y_1,\quad t_n=2y_{2n-1}\quad (n=1,2,\ldots),\\
&Z=\hat{w}^{(\infty)}_1(s=0;x,y)=-\frac{\partial}{\partial x_1}\log\tau(0;x,y),\\
&T=\frac{1}{2}c_{-1}(s=0;x,y)=\frac{\tau(1;x,y)^2}{2\tau(0;x,y)^2},\\
&\kappa = \sqrt{-1}\,c_0(s=0;x,y)=
\sqrt{-1}\,\frac{\partial}{\partial y_1}\log\frac{\tau(0;x,y)}{\tau(1;x,y)},\\
&f^{(n)}=-\sqrt{-1}\,F^{(n-1)}(s=0)\,, \quad
 g^{(n)}=G^{(n-1)}(s=0), 
\end{aligned}
\end{equation}
then $Z,T,\kappa,f^{(n)},g^{(n)}$ solve the equations
\eqref{def:coeff_f_g}, \eqref{eqs:GPhier_cplx_coord}.
\end{thm}
\begin{proof}
The first equation of \eqref{eqs:GPhier_cplx_coord} follows from 
\eqref{dw1/dy(2n+1)}. 
The second and the third are obtained from 
\eqref{dc(-1)/dy(2n+1)}, \eqref{dw1/dy_2n+1}.
The recurrence relations \eqref{def:coeff_f_g} follows 
from \eqref{dc2n/y1}, \eqref{def:FG} and \eqref{dF/dy1,dG/dy1}.
\end{proof}

\subsection{Discrete mKdV flow on discrete curves}
We recall a discrete analogue of the mKdV-flow of plane curve 
introduced by Matsuura \cite{Matsuura}. 
Let $\gamma^m_n:\mathbb{Z}^2\to\mathbb{C}$ be a map describing
the discrete motion of discrete plane curve with segment length $a_n$:
\begin{equation}
\begin{aligned}
& \left|\frac{\gamma^m_{n+1}-\gamma^m_n}{a_n}\right|=1,\\
& \frac{\gamma^m_{n+1}-\gamma^m_n}{a_n} = e^{\sqrt{-1}K^m_n}
\frac{\gamma^m_{n}-\gamma^m_{n-1}}{a_{n-1}}, \\
& \frac{\gamma^{m+1}_n-\gamma^m_n}{b_n} = e^{\sqrt{-1}W^m_n}
\frac{\gamma^m_{n+1}-\gamma^m_n}{a_n}.
\end{aligned}
\label{eq:dmKdVmotion}
\end{equation}
The compatibility condition for \eqref{eq:dmKdVmotion} implies 
the existence of the function $\theta^m_n$ defined by
\begin{equation}
W^m_n=\frac{\theta^{m+1}_n-\theta^m_{n+1}}{2},\quad
K^m_n=\frac{\theta^m_{n+1}-\theta^m_{n-1}}{2}.
\label{def:theta^m_n}
\end{equation}
Then the isoperimetric condition (the first equation in
\eqref{eq:dmKdVmotion}) 
implies that $\theta^m_n$ satisfies the 
discrete potential mKdV equation \cite{Hirota}
\begin{equation}
\tan\!\left(\frac{\theta^{m+1}_{n+1}-\theta^m_n}{2}\right)
=\frac{b_m+a_n}{b_m-a_n}
\tan\!\left(\frac{\theta^{m+1}_{n}-\theta^m_{n+1}}{2}\right).
\end{equation}

In what follows, we will show that the equations \eqref{eq:dmKdVmotion}
can be obtained from the Toda lattice hierarchy.
We introduce discrete variables $m,n\in\mathbb{Z}$ and assume 
$\tilde{y}_k$ depends on $m,n$ as 
\begin{equation}
\tilde{y}_k(m,n) = -\sum_{n'}^{n-1}\frac{a_{n'}^k}{k}
-\sum_{m'}^{m-1}\frac{b_{m'}^k}{k}
\quad (k=1,2,3,\ldots),
\label{def:naMiwaTrf}
\end{equation}
which is a non-autonomous version of Miwa transformation \cite{WTS}. 
We remark that if $a_n=a$ and $b_m=b$ for any $n,m$ then 
\eqref{def:naMiwaTrf} is reduced to original Miwa transformation 
\cite{Miwa}:
\begin{equation}
\tilde{y}_k(m,n) = -\frac{na^k}{k}-\frac{mb^k}{k}
\quad (k=1,2,3,\ldots).
\end{equation}
To consider the dependence on $m,n$, we use the following 
abbreviation:
\begin{equation}
\begin{aligned}
\hat{W}^{(\infty)}(s;m,n)&=\hat{W}^{(\infty)}(s;x,y=\tilde{y}(m,n)),\\
\hat{W}^{(0)}(s;m,n)&=\hat{W}^{(0)}(s;x,y=\tilde{y}(m,n)).
\end{aligned}
\end{equation}
\begin{prop}
$\hat{W}^{(\infty)}(s;m,n)$ and $\hat{W}^{(0)}(s;m,n)$ satisfy 
\begin{equation}
\begin{aligned}
&\hat{W}^{(\infty)}(s;m,n+1)
= \left\{1-a_n\tilde{u}(s;m,n)e^{-\partial_s}\right\}
\hat{W}^{(\infty)}(s;m,n),
\\
&\hat{W}^{(0)}(s;m,n+1)
\left(1-a_ne^{-\partial_s}\right)
= \left\{1-a_n\tilde{u}(s;m,n)e^{-\partial_s}\right\}
\hat{W}^{(0)}(s;m,n),
\\
&\hat{W}^{(\infty)}(s;m+1,n)
= \left\{1-b_m\tilde{v}(s;m,n)e^{-\partial_s}\right\}
\hat{W}^{(\infty)}(s;m,n),
\\
&\hat{W}^{(0)}(s;m+1,n)
\left(1-b_me^{-\partial_s}\right)
= \left\{1-b_m\tilde{v}(s;m,n)e^{-\partial_s}\right\}
\hat{W}^{(0)}(s;m,n),
\end{aligned}
\label{eq:discreteSato}
\end{equation}
where 
\begin{equation}
\begin{aligned}
\tilde{u}(s;m,n)=\frac{\hat{w}_0^{(0)}(s;m,n+1)}{\hat{w}_0^{(0)}(s-1;m,n)}
&=\frac{\tau(s-1;m,n)\tau(s+1;m,n+1)}{\tau(s;m,n)\tau(s;m,n+1)},
\\
\tilde{v}(s;m,n)=\frac{\hat{w}_0^{(0)}(s;m+1,n)}{\hat{w}_0^{(0)}(s-1;m,n)}
&=\frac{\tau(s-1;m,n)\tau(s+1;m+1,n)}{\tau(s;m,n)\tau(s;m+1,n)}.
\end{aligned}
\label{def:tildeU_tildeV}
\end{equation}
\end{prop}
\begin{proof}
Setting $x'_k=x_k$, $y'_k=\tilde{y}(m,n+1)$, 
$y_k=\tilde{y}(m,n)$ ($k=1,2,\ldots$) in \eqref{WexpW=WexpW}, 
we have
\begin{equation}
\begin{aligned}
&\hat{W}^{(\infty)}(s;m,n+1)\hat{W}^{(\infty)}(s;m,n)^{-1}\\
&\qquad =
\hat{W}^{(0)}(s;m,n+1)
\left(1-a_ne^{-\partial_s}\right)
\hat{W}^{(0)}(s;m,n)^{-1},
\end{aligned}
\label{WexpW_discrete_n}
\end{equation}
where we have used the formula 
$\exp\!\left(-\sum_{n=0}^{\infty}z^n/n\right)=1-z$.
Since the left-hand side of \eqref{WexpW_discrete_n} is of non-positive order
with respect to $e^{\partial_s}$, it follows that 
it is of the form
\begin{equation}
\eqref{WexpW_discrete_n} = 
\tilde{c}_0(s;m,n)+\tilde{c}_{-1}(s;m,n)e^{-\partial_s}.
\label{explicitForm_WexpW_discrete_n}
\end{equation}
Inserting $\hat{W}^{(\infty)}$ and $\hat{W}^{(0)}$ of 
\eqref{def:WinftyW0} to \eqref{WexpW_discrete_n} with 
\eqref{explicitForm_WexpW_discrete_n}, we obtain the 
first and the second equation of \eqref{eq:discreteSato}. 
The third and the fourth can be obtained in the same fashion.
\end{proof}
\noindent\textbf{Remark:} 
Tsujimoto \cite{Tsujimoto} proposed and investigated the equations 
\eqref{eq:discreteSato} as a discrete analogue of \eqref{SatoEqs}. 
In our approach, the results in \cite{Tsujimoto} 
can be obtained directly from \eqref{WexpW=WexpW} with 
the Miwa transformation.
\medskip

Hereafter in this section, we impose the 2-reduction condition 
$\tau(s+2;m,n)=\tau(s;m,n)$. 
{}From the first and the third equations of \eqref{eq:discreteSato}, 
we obtain
\begin{equation}
\begin{aligned}
\hat{w}^{(\infty)}_1(s;m,n+1)
&=\hat{w}^{(\infty)}_1(s;m,n)-a_n\tilde{u}(s;m,n),
\\
\hat{w}^{(\infty)}_1(s;m+1,n)&=
\hat{w}^{(\infty)}_1(s;m,n)-b_m\tilde{v}(s;m,n).
\end{aligned}
\label{discreteEqs_w1}
\end{equation}
It follows that
\begin{equation}
\begin{aligned}
\lefteqn{\frac{\hat{w}^{(\infty)}_1(s;m,n+1)-
\hat{w}^{(\infty)}_1(s;m,n)}{a_n}}\qquad\\
&=\mathcal{K}(s;m,n)
\frac{\hat{w}^{(\infty)}_1(s;m,n)-
\hat{w}^{(\infty)}_1(s;m,n-1)}{a_{n-1}},
\\[2mm]
\lefteqn{\frac{\hat{w}^{(\infty)}_1(s;m+1,n)-
\hat{w}^{(\infty)}_1(s;m,n)}{b_m}}\qquad\\
&=\mathcal{W}(s;m,n)
\frac{\hat{w}^{(\infty)}_1(s;m,n+1)-
\hat{w}^{(\infty)}_1(s;m,n)}{a_n},
\end{aligned}
\label{w1_K_W}
\end{equation}
with
\begin{equation}
\begin{aligned}
\mathcal{K}(s;m,n)&=\frac{\tilde{u}(s;m,n)}{\tilde{u}(s;m,n-1)}
=\frac{\tau(s+1;m,n+1)\tau(s;m,n-1)}{\tau(s;m,n+1)\tau(s+1;m,n-1)},
\\
\mathcal{W}(s;m,n)&=\frac{\tilde{v}(s;m,n)}{\tilde{u}(s;m,n)}
=\frac{\tau(s+1;m+1,n)\tau(s;m,n+1)}{\tau(s;m+1,n)\tau(s+1;m,n+1)}.
\end{aligned}
\end{equation}
If we introduce $\Theta(s;m,n)$ as 
\begin{equation}
\Theta(s;m,n)= \tau(s+1;m,n)/\tau(s;m,n),
\end{equation}
then $\mathcal{K}(s;m,n)$ and $\mathcal{W}(s;m,n)$ are written as
\begin{equation}
\mathcal{K}(s;m,n)=\frac{\Theta(s;m,n+1)}{\Theta(s;m,n-1)},\quad
\mathcal{W}(s;m,n)=\frac{\Theta(s;m+1,n)}{\Theta(s;m,n+1)}.
\label{K_W_Theta}
\end{equation}

We furthermore impose the reality condition \eqref{reality_condition}.
Under the condition, $\Theta(s;m,n)$ satisfies 
$\left|\Theta(s;m,n)\right|=1$ and one can set 
\begin{equation}
e^{\sqrt{-1}\theta^m_n}=\Theta(s=0;m,n)=\tau(1;m,n)/\tau(0;m,n).
\label{theta_Theta}
\end{equation}
\begin{thm}[Representation formula for discrete curves 
in terms of the $\tau$-functions]
If we set 
\begin{equation}
\begin{aligned}
\gamma^m_n &= \hat{w}^{(\infty)}_1(s=0;m,n)
=-\frac{\partial}{\partial x_1}\log\tau(0;m,n),\\
\theta^m_n &= \frac{1}{\sqrt{-1}}\log\Theta(s=0;m,n)
= \frac{1}{\sqrt{-1}}\log\frac{\tau(1;m,n)}{\tau(0;m,n)},
\end{aligned}
\end{equation}
then $\gamma^m_n$ and $\theta^m_n$ solve the equations 
\eqref{eq:dmKdVmotion} and \eqref{def:theta^m_n}.
\end{thm}
\begin{proof}
{}From \eqref{discreteEqs_w1} and \eqref{def:tildeU_tildeV}, 
it follows that 
\begin{equation}
\begin{aligned}
&\left|\frac{\hat{w}^{(\infty)}_1(s;m,n+1)-
\hat{w}^{(\infty)}_1(s;m,n)}{a_n}\right|
\\
&\qquad 
=\left|
\frac{\tau(s-1;m,n)\tau(s+1;m,n+1)}{\tau(s;m,n)\tau(s;m,n+1)}
\right|=1
\end{aligned}
\end{equation}
under the condition \eqref{reality_condition}. 
This is equivalent to the first equation of \eqref{eq:dmKdVmotion}. 
The remaining equations follow directly from 
\eqref{w1_K_W}, \eqref{K_W_Theta} and \eqref{theta_Theta}.
\end{proof}

\section{Fermionic construction of $\tau$-functions}
In \cite{Takebe1,Takebe2}, Takebe described $\tau$-functions for 
the Toda hierarchy as expectation values of fermionic operators
(See also \cite{TakasakiBook}). 
We firstly recall the definition of charged free fermions 
\cite{JimboMiwa,MJD}.

Let $\mathcal{A}$ be an associative unital $\mathbb{C}$-algebra 
generated by $\psi_i$, $\psi^*_i$ ($i\in\mathbb{Z}$) 
satisfying the relations
\begin{equation}
\psi_i\psi^*_j+\psi^*_j\psi_i=\delta_{ij},\quad
\psi_i\psi_j+\psi_j\psi_i=
\psi^*_i\psi^*_j+\psi^*_j\psi^*_i=0.
\end{equation}
We consider a class of infinite matrices
$A=\left[a_{ij}\right]_{i,j\in\mathbb{Z}}$ that satisfies the 
following condition:
\begin{equation}
\mbox{there exists }N>0\mbox{ such that }a_{ij}=0
\mbox{ for all }i,j\mbox{ with }|i-j|>N.
\label{gl_infty_condition}
\end{equation}
Define the Lie algebra $\mathfrak{gl}(\infty)$ as \cite{JimboMiwa}
\begin{equation}
\mathfrak{gl}(\infty)=
\left\{
\sum_{i,j\in\mathbb{Z}}a_{ij} :\!\psi_i\psi^*_j\!:\Bigg|
\;A=\left[a_{ij}\right]_{i,j\in\mathbb{Z}}
\mbox{ satisfies \eqref{gl_infty_condition}}
\right\}\oplus\mathbb{C}
\end{equation}
where $:\cdot :$ indicates the normal ordering
\begin{equation}
\begin{aligned}
:\!\psi_i\psi^*_j\!: \, = 
\begin{cases}
\psi_i\psi^*_j & \mbox{if $i\neq j$ or $i=j\geq 0$},\\
-\psi^*_j\psi_i & \mbox{if $i=j<0$}.
\end{cases}
\end{aligned}
\end{equation}
We also define the group $\mathbf{G}$ corresponds to
$\mathfrak{gl}(\infty)$ to be
\begin{equation}
\mathbf{G}=\left\{
e^{X_1}e^{X_2}\dots e^{X_k}\;\big|\; 
X_i\in\mathfrak{gl}(\infty)\right\}.
\end{equation}

Consider a left $\mathcal{A}$-module with a cyclic vector
$|\mathrm{vac}\rangle$ satisfying
\begin{equation}
\psi_j|\mathrm{vac}\rangle=0 \quad (j<0),\quad
\psi^*_k|\mathrm{vac}\rangle=0 \quad (k\geq 0).
\end{equation}
The $\mathcal{A}$-module $\mathcal{A}|\mathrm{vac}\rangle$ is called 
the fermion Fock space $\mathcal{F}$, which we denote $\mathcal{F}$. 
We also consider a right $\mathcal{A}$-module (the dual Fock space 
$\mathcal{F}^*$) with a cyclic vector $\langle\mathrm{vac}|$ satisfying
\begin{equation}
\langle\mathrm{vac}|\psi_j=0 \quad (j\geq 0), \quad
\langle\mathrm{vac}|\psi^*_k=0 \quad (k<0).
\end{equation}
We further define the generalized vacuum vectors 
$|s\rangle$, $\langle s|$ ($s\in\mathbb{Z}$) as 
\begin{equation}
\begin{aligned}
|s\rangle &= 
\begin{cases}
\psi^*_s\cdots\psi^*_{-1}|\mathrm{vac}\rangle & \mbox{for }s<0,\\
|\mathrm{vac}\rangle & \mbox{for }s=0,\\
\psi_{s-1}\cdots\psi_{0}|\mathrm{vac}\rangle & \mbox{for }s>0,
\end{cases}
\\
\langle s| &= 
\begin{cases}
\langle\mathrm{vac}|\psi_{-1}\cdots\psi_s & \mbox{for }s<0,\\
\langle\mathrm{vac}| & \mbox{for }s=0,\\
\langle\mathrm{vac}|\psi^*_0\cdots\psi^*_{s-1} & \mbox{for }s>0.
\end{cases}
\end{aligned}
\end{equation}

There exists a unique linear map (the vacuum expectation value) 
$\mathcal{F}^*\otimes_{\mathcal{A}}\mathcal{F}\to\mathbb{C}$ such that
$\langle\mathrm{vac}|\otimes |\mathrm{vac}\rangle \mapsto 1$
For $a\in\mathcal{A}$ we denote by 
$\langle\mathrm{vac}|a|\mathrm{vac}\rangle$
the vacuum expectation value of the vector
$\langle\mathrm{vac}|a\otimes |\mathrm{vac}\rangle=
\langle\mathrm{vac}|\otimes a|\mathrm{vac}\rangle$
in $\mathcal{F}^*\otimes_{\mathcal{A}}\mathcal{F}$.

\begin{thm}[\cite{Takebe1} \S 2, \cite{Takebe2} \S 2]
For $s\in\mathbb{Z}$ and $g\in\mathbf{G}$, define 
$\tau_g(s;x,y)$ as
\begin{equation}
\tau_g(s;x,y)=
\langle s|e^{H(x)}ge^{-\bar{H}(y)}|s\rangle, \quad
\label{fermion_tau}
\end{equation}
where
\begin{equation}
H(x)=\sum_{n=1}^{\infty}x_n\sum_{j\in\mathbb{Z}}\psi_j\psi^*_{j+n}, \quad
\bar{H}(y)=\sum_{n=1}^{\infty}y_n
\sum_{j\in\mathbb{Z}}\psi_{j+n}\psi^*_j.
\end{equation}
Then $\tau_g(s;x,y)$ satisfies the bilinear identity \eqref{BilinearIdentity}.
\end{thm}

We introduce an automorphism $\iota_l$ of $\mathcal{A}$ by
\begin{equation}
\iota_l(\psi_i)=\psi_{i-l},\quad
\iota_l(\psi^*_i)=\psi^*_{i-l}, 
\end{equation}
which satisfies
\begin{equation}
\langle s'|a|s\rangle=\langle s'-l|\iota_l(a)|s-l\rangle
\label{<s|iota(a)|s>}
\end{equation}
for any $s,s',l$ and any $a\in\mathcal{A}$.
\begin{prop}
If $g\in\mathbf{G}$ satisfies
\begin{equation}
\iota_1(g) = \overline{g},
\label{realityCondition_fermions}
\end{equation}
then the $\tau$-function corresponds to $g$ gives a solution 
of the Goldstein-Petrich hierarchy.
\end{prop}
\begin{proof}
{}From \eqref{<s|iota(a)|s>} and \eqref{realityCondition_fermions}, 
it is clear that \eqref{reality_condition} holds.
\end{proof}

To construct soliton-type solutions, we choose $g$ as 
\begin{equation}
\begin{aligned}
& g_N(\{c_j\},\{p_j\},\{q_j\})=\prod_{j=1}^N
 e^{c_i \psi(p_i)\psi^*(q_i)},\\
& \psi(p)=\sum_{j\in\mathbb{Z}}\psi_jp^j,
\quad \psi^*(q)=\sum_{j\in\mathbb{Z}}\psi^*_jq^{-j}.
\end{aligned}
\label{Nsoliton-tau}
\end{equation}
We remark that the vacuum expectation value of 
$e^{c\psi(p)\psi^*(q)}$
makes sense even when
$X=c\psi(p)\psi^*(q)$  does not satisfy the condition
\eqref{gl_infty_condition}:
\begin{equation}
\langle s|e^{c\psi(p)\psi^*(q)}|s\rangle
=\langle s|\left\{1+c\psi(p)\psi^*(q)\right\}
|s\rangle
=1+\left(\frac{p}{q}\right)^s\frac{cq}{p-q}.
\end{equation}

We consider the following two types of conditions for the parameters in
\eqref{Nsoliton-tau}:
\begin{itemize}
\item[A.] (Soliton solutions)
\begin{equation}
c_j\in\sqrt{-1}\mathbb{R}, \quad p_j\in\mathbb{R}, \quad
q_j=-p_j \quad (j=1,2,\dots,N),
\label{ConditionA}
\end{equation}
\item[B.] (Breather solutions)
\begin{equation}
\begin{aligned}
& N=2M, \quad \overline{c_{2k-1}}=-c_{2k}, \quad 
\overline{p_{2k-1}}=p_{2k} \quad (k=1,2,\dots,M),\\
& q_j=-p_j \quad (j=1,2,\dots,N).
\end{aligned}
\label{ConditionB}
\end{equation}
\end{itemize}
An straightforward calculation shows that 
$g_N\left(\{c_j\},\{p_j\},\{q_j\}\right)$ satisfies 
\eqref{realityCondition_fermions} under each of the conditions 
\eqref{ConditionA}, \eqref{ConditionB}.
The $\tau$-functions under these conditions provide 
the solutions given in \cite{IKMO1,IKMO2}.

We now consider Lie algebraic meaning of the condition 
\eqref{realityCondition_fermions}. We recall the facts about 
a fermionic representation of the affine Lie algebra
$\widehat{\mathfrak{sl}}(2,\mathbb{C})$. 
The affine Lie algebra
$\widehat{\mathfrak{sl}}(2,\mathbb{C})$ 
is generated by the Chevalley generators 
$\left\{e_0,e_1,f_0,f_1,h_0,h_1\right\}$ 
that satisfy
\begin{equation}
\begin{aligned}
&\left[h_i,h_j\right]=0,\quad 
\left[e_i,f_j\right]=\delta_{ij}h_i \mbox{ \ for all \ }i,j,\\
&\left[h_i,e_j\right]=
\begin{cases}
2e_j & \mbox{if \ }i=j,\\
-2e_j & \mbox{if \ }i\neq j,
\end{cases}\qquad
\left[h_i,f_j\right]=
\begin{cases}
-2e_j & \mbox{if \ }i=j,\\
2e_j & \mbox{if \ }i\neq j,
\end{cases}\\
& \left[e_i,\left[e_i,\left[e_i,e_j\right]\right]\right]=
\left[f_i,\left[f_i,\left[f_i,f_j\right]\right]\right]=0
\mbox{ \ if \ }i\neq j.
\end{aligned}
\end{equation}
Define a linear map 
$\pi:\widehat{\mathfrak{sl}}(2,\mathbb{C})\to\mathfrak{gl}(\infty)$ as 
\begin{equation}
\begin{aligned}
\pi\left(e_j\right) 
&= \sum_{n\equiv j\;\mathrm{mod}\;2}\psi_{n-1}\psi^*_{n},\quad
\pi\left(f_j\right)
 = \sum_{n\equiv j\;\mathrm{mod}\;2}\psi_{n}\psi^*_{n-1},\\
\pi\left(h_j\right) &= \sum_{n\equiv j\;\mathrm{mod}\;2}
\left(\,:\!\psi_{n-1}\psi^*_{n-1}\!:\, -
\,:\!\psi_{n}\psi^*_{n}\!:\, \right) + \delta_{j0} \quad (j=0,1).
\end{aligned}
\end{equation}
\begin{thm}[\cite{JimboMiwa,MJD}]
$\left(\pi,\mathcal{F}\right)$ is a representation of 
$\widehat{\mathfrak{sl}}(2,\mathbb{C})$.
\end{thm}

Note that $\iota_1$ works as an involutive automorphism:
\begin{equation}
\iota_1(e_0) = e_1, \quad \iota_1(f_0) = f_1, \quad
\iota_1(e_1) = e_0, \quad \iota_1(f_1) = f_0, 
\end{equation}
which defines a real form of $\widehat{\mathfrak{sl}}(2,\mathbb{C})$.
Kobayashi \cite{Z.Kobayashi} classified automorphisms of 
prime order of the affine Lie algebra
$\widehat{\mathfrak{sl}}(n,\mathbb{C})$. 
The involutive automorphism 
$\iota_1$ under consideration is labeled as (1a')-type
(\cite{Z.Kobayashi}, Theorem 3). 
We remark that the same real form of
$\widehat{\mathfrak{sl}}(2,\mathbb{C})$ appeared also in 
construction of solutions of a derivative nonlinear Sch\"odinger 
equation \cite{KIT}.

\def\theequation{A.\arabic{equation}}
\setcounter{equation}{0}
\section*{Appendix: Time-flows with positive weight}
So far, we have used the time-evolutions with respect to
the variables with negative weight $y=(y_1,y_2,\ldots)$ to derive the
Goldstein-Petrich hierarchy. 
In this appendix, we use $x=(x_1,x_2,\ldots)$ and show that 
the mKdV hierarchy can be obtained under the 2-reduction condition
\eqref{def:2-reduction}. 
Applying the condition \eqref{def:2-reduction}, one can show that 
\begin{equation}
\begin{aligned}
B_{2n-1}(s)&
=e^{(2n-1)\partial_s} + \sum_{-2(n-1)\leq j\leq 0}
 b_j(s)e^{(2n-2+j)\partial_s},
\\
B_{2n}(s)&= e^{2n\partial_s} \quad (n=1,2,\ldots).
\end{aligned}
\label{2-reduced_B}
\end{equation}
{}From \eqref{def:L} and \eqref{def:2-reduction}, we obtain
\begin{equation}
\begin{aligned}
& b_0(s+1)+b_0(s)=0,\\
& b_{-k-1}(s+1)+b_{-k-1}(s)
+\sum_{j=0}^k b_{-j}(s)b_{j-k}(s-j)=0
\quad (k=0,1,2,\ldots).
\end{aligned}
\end{equation}
Applying \eqref{2-reduced_B} to \eqref{def:TLhierarchy_x}, 
we obtain
\begin{equation}
\frac{\partial b_0(s)}{\partial x_{2n-1}}
=b_{-2n+1}(s+1)-b_{-2n+1}(s).
\label{DifEq:db0/dt(2n-1)}
\end{equation}

Define $L_1(x,y)$, $L_2(x,y)$ by
\begin{equation}
\begin{aligned}
L_1(x,y)&=\frac{1}{2}\left\{
L^{(\infty)}(s=0;x,y)-L^{(\infty)}(s=1;x,y)\right\},\\
L_2(x,y)&=\frac{1}{2}\left\{
L^{(\infty)}(s=0;x,y)+L^{(\infty)}(s=1;x,y)\right\},
\end{aligned}
\end{equation}
which have the following form:
\begin{equation}
\begin{aligned}
L_1(x,y)&=\sum_{n=0}^{\infty} q_n(x,y) e^{-n\partial_s},\quad
L_2(x,y)=e^{\partial_s} + \sum_{n=1}^{\infty}
r_n(x,y)e^{-n\partial_s},\\
q_n(x,y) &= \frac{b_{-n}(s=0,x,y)-b_{-n}(s=1,x,y)}{2}
\quad (n=0,1,2,\ldots), \\
r_n(x,y) &= \frac{b_{-n}(s=0,x,y)+b_{-n}(s=1,x,y)}{2}
\quad (n=1,2,3,\ldots).
\end{aligned}
\label{def:q,r}
\end{equation}
We remark that $q_n$ and $r_n$ are eigenfunctions of $e^{\partial_s}$:
\begin{equation}
e^{\partial_s}q_n = -q_n, \quad
e^{\partial_s}r_n = r_n.
\end{equation}
Applying the notation \eqref{def:q,r} to \eqref{DifEq:db0/dt(2n-1)}, 
we have
\begin{equation}
\frac{\partial q_0}{\partial x_{2n-1}}
=-2q_{2n-1}
\label{DiffEq:dq0/dx(2n-1)}
\end{equation}

Since $B_1(0)$, $B_1(1)$ are of the form
\begin{equation}
B_1(0)=e^{\partial_s}+q_0,\quad 
B_1(1)=e^{\partial_s}-q_0,
\end{equation}
it follows that
\begin{equation}
\frac{\partial L_1}{\partial x_1} =
-2L_1 e^{\partial_s} + \left[q_0,\,L_2\right],
\quad 
\frac{\partial L_2}{\partial x_1} =
\left[q_0,\,L_1\right],
\end{equation}
and hence
\begin{equation}
\begin{aligned}
\frac{\partial q_{2n-1}}{\partial x_1} &=-2 q_{2n}+2q_0r_{2n-1},&
\frac{\partial q_{2n}}{\partial x_1} &=-2 q_{2n+1},\\
\frac{\partial r_{2n-1}}{\partial x_1} &=2 q_0q_{2n-1}, &
\frac{\partial r_{2n}}{\partial x_1} &= 0.
\end{aligned}
\label{DiffEqs:dq/dx1,dr/dx1}
\end{equation}
{}From \eqref{DiffEq:dq0/dx(2n-1)} and \eqref{DiffEqs:dq/dx1,dr/dx1}, 
we have
\begin{equation}
\frac{\partial q_0}{\partial x_{2n+1}}=\left(
\frac{1}{4}\partial_{x_1}^2-q_0^2-\frac{\partial q_0}{\partial x_1}
\partial_{x_1}^{-1}\circ q_0
\right) \frac{\partial q_0}{\partial x_{2n-1}}.
\label{recursion_q}
\end{equation}
Especially for the case $n=1$, 
\begin{equation}
\frac{\partial q_0}{\partial x_3}
=\frac{1}{4}\frac{\partial^3 q_0}{\partial x_1^3}
-\frac{3}{2}q_0^2\frac{\partial q_0}{\partial x_1}.
\label{mKdVeq_461}
\end{equation}
After suitable scaling, 
the linear operator appeared in the right-hand side of 
\eqref{recursion_q} yields the recursion operator 
$\Omega$ in \eqref{def:recursion_operator_omega}, and 
the equation \eqref{mKdVeq_461} 
yields the mKdV equation \eqref{mKdV_kappa}.

We remark that another derivation of the recursion operator $\Omega$ in
terms of bilinear differential equations of Hirota-type was given in
\cite{WLS}. Here we briefly summarize the approach in \cite{WLS}. 
We use the Hirota differential operators $D_x$, $D_y$, $\ldots$, defined by
\begin{equation}
\begin{aligned}
D_x^m D_y^n f(x,y)\cdot g(x,y)
=\left.\left(\partial_x-\partial_{x'}\right)^m
\left(\partial_x-\partial_{x'}\right)^n
f(x,y)g(x',y')\right|_{x'=x,y'=y}.
\end{aligned}
\end{equation}
Setting $s'=0$, $s=1$ $y'_n=y_n$, $x'_n=x_n+a_n$ ($n=1,2,\ldots$), 
the bilinear identity \eqref{BilinearIdentity} is reduced to 
\begin{equation}
\oint\tau(0;x'-[\lambda^{-1}],y)\tau(1;x+[\lambda^{-1}],y)
e^{\xi(x'-x,\lambda)}\lambda^{-1}d\lambda
=\tau(1;x',y)\tau(0;x,y),
\end{equation}
or, using the Hirota operators 
$\tilde{D}=(D_{1},D_{2}/2,D_{3}/3,\ldots)$, 
$D_j=D_{x_j}$ ($j=1,2,\ldots$), we can write
\begin{equation}
\sum_{j=0}^{\infty}p_j(-2a)p_j(\tilde{D})
\exp\!\left(\sum_{k=1}^{\infty}a_kD_k\right)
\tau(0)\cdot\tau(1)
=\exp\!\left(\sum_{k=1}^{\infty}a_kD_k\right)
\tau(1)\cdot\tau(0)
\label{canonical_BilinearIdentity}
\end{equation}
for any $a=(a_1,a_2,\ldots)$ (cf. \cite{Loris}).
Expanding \eqref{canonical_BilinearIdentity} with respect to the 
variables $a=(a_1,a_2,\ldots)$, we obtain
\begin{equation}
\left(p_m(\tilde{D})-D_m\right)\tau(1)\cdot\tau(0)=0
\label{bilinear:pm(tilD)-Dm}
\end{equation}
from the coefficient of $a_m$, and 
\begin{equation}
\left(-2p_{m+k}(\tilde{D})+p_m(\tilde{D})D_k+p_k(\tilde{D})D_m
\right)\tau(1)\cdot\tau(0)=0
\label{bilinear:-2p(m+k)}
\end{equation}
from the coefficient of $a_m a_k$. 
Using \eqref{bilinear:pm(tilD)-Dm} to eliminate the first term 
in \eqref{bilinear:-2p(m+k)}, we have
\begin{equation}
\left(-2D_{m+k}+p_m(\tilde{D})D_k+p_k(\tilde{D})D_m
\right)\tau(1)\cdot\tau(0)=0.
\label{bilinear:-2D(m+k)}
\end{equation}

Hereafter we impose the 2-reduction condition $\partial_{x_{2n}}\tau=0$ 
($n=1,2,\ldots$). 
Setting $k=2$, the bilinear equations \eqref{bilinear:pm(tilD)-Dm}, 
\eqref{bilinear:-2D(m+k)} yield 
\begin{equation}
D_{1}^2 \tau(1)\cdot\tau(0)=0,\quad
\left(-4D_{m+2}+D_{1}^2D_{m}\right)\tau(1)\cdot\tau(0)=0.
\label{2-reduced_bilinearEqs}
\end{equation}
If we set
\begin{equation}
\psi= \log\left(\tau(1)/\tau(0)\right), \quad \phi=\log\left(\tau(0)\tau(1)\right), 
\end{equation}
it follows that
\begin{equation}
(\partial_{1}\psi)^2 + \partial_{1}^2\phi =0, \quad
-4\partial_{m+2}\psi +\partial_{1}^2\partial_{m}\psi
+2(\partial_{1}\psi)(\partial_{1}\partial_{m}\phi)=0,
\end{equation}
{}from \eqref{2-reduced_bilinearEqs}, 
where $\partial_{n}=\partial/\partial x_n$.
Setting 
\begin{equation}
q_0=\partial_1\psi=\partial_1\left(\log\frac{\tau(1)}{\tau(0)}\right), 
\end{equation}
we have the recursion relation \eqref{recursion_q}.

\section*{Acknowledgments}
S.K. acknowledges Simpei Kobayashi for bringing his attention to
the paper \cite{Z.Kobayashi}, and Ralph Willox for explaining the 
results of \cite{WLS}.
This work is partially supported by JSPS Grant-in-Aid for 
Scientific Research No. 23340037 and No. 23540252.

\end{document}